\documentclass[a4paper,12pt]{article}

\usepackage{amsmath,amssymb,amsthm}
\usepackage[top=2cm,nohead,left=2.5cm,right=2.5cm]{geometry}
\usepackage{setspace}
\usepackage{natbib}
\usepackage{hyperref}
\hypersetup{hidelinks}

\newtheorem{theorem}{Theorem}
\newtheorem{lemma}{Lemma}
\newtheorem{proposition}{Proposition}
\newtheorem{corollary}{Corollary}
\newtheorem{assumption}{Assumption}
\theoremstyle{definition}
\newtheorem{definition}{Definition}

\theoremstyle{remark}

\begin{document}

\title{Doctor-Optimal Stability in Unitary Many-to-Many Markets}
\author{Yi-You Yang\thanks{Department of Applied Mathematics, Chung Yuan Christian University, Taoyuan City, Taiwan. E-mail address: yyyang@cycu.edu.tw}}
\maketitle

\begin{abstract}
We study bilaterally unitary many-to-many doctor--hospital matching with
contracts, taking choice functions as primitives. Doctor choices are
substitutable and satisfy irrelevance of rejected contracts, while hospital
choices are unilaterally substitutable and satisfy the same condition. Every
trajectory of the doctor-proposing cumulative offer process terminates at the
greatest stable allocation under the doctor Blair order. We also introduce
weakly hospital-quasi-stable allocations and show that they form a finite
lattice whose greatest element is stable. Hence, the cumulative-offer outcome,
the greatest weakly hospital-quasi-stable allocation, and the doctor-optimal
stable allocation coincide. The common allocation is hospital-pessimal in the
revealed-choice sense. Under the law of aggregate demand, every agent signs
the same number of contracts at all stable allocations.
\end{abstract}

\section{Introduction}
\label{sec:introduction}

Matching with contracts allows agents to choose both their partners and the
terms governing their relationships \citep{Hatfield2005a}. We study
bilaterally unitary many-to-many doctor--hospital markets. Each doctor and
hospital may maintain several relationships, but a doctor--hospital pair may
sign at most one contract. Contracts involving the same pair can therefore be
viewed as mutually exclusive specifications of a single bilateral
relationship. The model permits portfolio choice across partners while
separating partner choice from the choice of contractual terms within a
relationship.

Our question is whether unilateral substitutability guarantees a
doctor-optimal stable allocation in this domain. Taking choice functions as
primitives, a stable allocation \(Y^D\) is doctor-optimal if
\[
C_D(Y^D\cup Y)=Y^D
\qquad
\text{for every stable allocation }Y.
\]
Thus, \(Y^D\) is the greatest stable allocation under the doctor Blair order.
When doctor choice functions are rationalizable by strict preferences over
bundles, this condition is equivalent to every doctor weakly preferring her
bundle at \(Y^D\) to her bundle at any other stable allocation.

\citet{Hatfield2010} establish doctor-optimal stability under unilateral
substitutability in the standard many-to-one model, where doctors are unit
demand. Their conclusion does not extend to general nonunitary many-to-many
markets. \citet{Bando2026} introduce substitutability across doctors, an
identity-level extension of unilateral substitutability, and show that it
guarantees the existence of a stable outcome when doctor choices are
responsive. Their Example~3 shows, however, that the stronger
doctor-optimal conclusion may fail. The hospital choice function in that
example is unilaterally substitutable and satisfies substitutability across
doctors, and the doctor choices are responsive. The cumulative offer process
has an order-independent stable outcome, but the market has two stable
outcomes and no doctor-optimal stable allocation.

We identify bilateral unitarity as a sufficient boundary for restoring the
doctor-optimal conclusion. Doctors have substitutable and unitary choice
functions satisfying irrelevance of rejected contracts. Hospitals have
unilaterally substitutable and unitary choice functions satisfying the same
consistency condition. Under hospital unitarity, unilateral substitutability
is equivalent to substitutability across doctors \citep{Bando2026}. We show
that these conditions guarantee a greatest stable allocation under the doctor
Blair order.

The comparison with \citet[Example~3]{Bando2026} isolates the role of
unitarity. The failure of doctor-optimality in their example is not caused by
portfolio choice across several partners: our model retains such portfolio
choice on both sides. The difference is that their nonunitary environment
allows several contracts to coexist within the same bilateral relationship.
Once alternative contractual specifications of a relationship are required
to be mutually exclusive, identity-level substitutability again supports a
doctor-optimal stable selection. Bilateral unitarity is therefore a
substantive domain restriction rather than a normalization of the model.

Our first characterization is procedural. We study a doctor-proposing
cumulative offer process in which one eligible contract is proposed at each
step and hospitals choose from the offers accumulated so far. Unilateral
substitutability does not imply monotonicity of hospital rejections on
arbitrary menus: a hospital may replace one contract with a doctor by another
contract involving the same doctor. We nevertheless show that hospital
rejections are irreversible along every cumulative-offer trajectory.
Hospital unitarity converts unilateral substitutability into preservation of
doctor identities, while doctor unitarity prevents two alternative contracts
within the same relationship from being simultaneously chosen. These
properties rule out a reversal of an earlier rejection.

It follows that every cumulative-offer trajectory terminates at a stable
allocation. More strongly, its terminal allocation \(Y^{COP}\) satisfies
\[
C_D(Y^{COP}\cup Y)=Y^{COP}
\qquad
\text{for every stable allocation }Y.
\]
Hence, \(Y^{COP}\) is the greatest stable allocation under the doctor Blair
order. All cumulative-offer trajectories therefore have the same outcome.
This strengthens order independence: in the nonunitary example of
\citet{Bando2026}, the process is order-independent but its outcome is not
doctor-optimal; under bilateral unitarity, every trajectory reaches the
greatest stable allocation.

The hospital-side condition is also sharp at the domain level.
\citet{Kasuya2021} show that unilateral substitutability is necessary for
guaranteed doctor-optimal stability within the observably
substitutable-across-doctors domain and is necessary more generally in the
maximal-domain sense. His counterexamples use unit-demand doctors and unitary
hospital choice functions satisfying irrelevance of rejected contracts.
They therefore lie within the choice-function domain considered here.
Combining his necessity result with our sufficiency theorem shows that,
within the unitary hospital domain, the class of unilaterally substitutable
choice functions is inclusion-maximal for a universal guarantee of
doctor-optimal stability. This is a domain-level statement; it does not imply
that every particular market violating unilateral substitutability lacks a
doctor-optimal stable allocation.

The preceding results reveal an asymmetric boundary. Doctor demand can be
expanded from unit demand to substitutable portfolio choice without losing
doctor-optimal stable selection. By contrast, the hospital-side requirement
of substitution across doctor identities cannot generally be weakened. In
this sense, our result connects the many-to-one sufficiency theorem of
\citet{Hatfield2010}, the maximal-domain necessity theorem of
\citet{Kasuya2021}, and the nonunitary failure documented by
\citet[Example~3]{Bando2026}.

Our second characterization is order-theoretic. It builds on
firm-quasi-stability, which relaxes stability by requiring firms to retain
their current contracts when choosing from all contracts acceptable to
workers \citep{Sotomayor1996,Yang2025,Yang2026}. Contract retention is too
strong under unilateral substitutability. A hospital may change the terms
associated with a doctor while retaining the underlying relationship. We
therefore introduce \emph{weak hospital-quasi-stability}, which replaces
contract retention with doctor-identity retention.

A weakly hospital-quasi-stable allocation is doctor-individually rational and
requires each hospital, when choosing from all contracts acceptable to
doctors relative to the current allocation, to retain some contract with
every doctor it currently employs. The hospital need not retain the same
contract. The notion thus describes a relationship-preserving renegotiation
state: contractual terms may change, while existing bilateral relationships
remain protected.

We prove that weakly hospital-quasi-stable allocations form a finite lattice
under the doctor Blair order. Let \(Y^Q\) be its greatest element. Define
\[
\Gamma(Y)
:=
\{x\in X:x\in C_D(Y\cup\{x\})\}
\]
and
\[
T(Y):=C_D(C_H(\Gamma(Y))).
\]
The operator maps the weakly hospital-quasi-stable domain into itself and
satisfies
\[
T(Y)\succeq_D^B Y.
\]
No isotonicity of \(T\) is required. We show that the fixed points of \(T\)
within the relaxed domain are precisely the stable allocations. Since
\(Y^Q\) is the greatest element of that domain, the improvement property
implies \(T(Y^Q)=Y^Q\), and hence \(Y^Q\) is stable.

The procedural and order-theoretic constructions coincide:
\[
Y^{COP}
=
Y^Q
=
\max_{\succeq_D^B}\mathcal S
=
Y^D.
\]
The cumulative offer process describes how the doctor-optimal stable
allocation is reached. The weakly hospital-quasi-stable lattice explains why
it exists even though the stable set itself need not form a lattice. The
construction identifies the order structure that survives when
contract-level substitutability is replaced by substitution across partner
identities.

The common allocation is also hospital-pessimal in a pairwise
revealed-choice sense. For every stable allocation \(Y\),
\[
C_H(Y^D\cup Y)=Y.
\]
Thus, when a hospital is offered its bundles at \(Y^D\) and \(Y\), it chooses
the latter. When hospital choice functions are rationalizable, this is the
usual hospital-pessimality conclusion. It identifies an ordinal boundary of
the stable set and does not impose a cardinal or zero-sum welfare
interpretation.

Finally, if every agent satisfies the law of aggregate demand, every doctor
and hospital signs the same number of contracts at all stable allocations.
Under bilateral unitarity, these contract counts equal the numbers of active
bilateral relationships. Stable allocations may differ in partners and
contractual terms, but not in the number of relationships maintained by each
agent.

The remainder of the paper is organized as follows. Section~\ref{sec:model}
introduces the model, the choice conditions, and stability.
Section~\ref{sec:COP} studies the cumulative offer process and establishes
doctor-optimality, order independence, and the maximality implication.
Section~\ref{sec:weak-quasi-stability} develops the weakly
hospital-quasi-stable lattice and the improvement operator.
Section~\ref{sec:equivalence} combines the two characterizations and derives
hospital pessimality and the rural hospitals theorem.
Section~\ref{sec:conclusion} concludes.

\section{Model and stability}
\label{sec:model}

\subsection{Environment and choice behavior}

A market consists of finite sets of doctors \(D\), hospitals \(H\), and
contracts \(X\). Each contract \(x\in X\) specifies a doctor \(d(x)\in D\)
and a hospital \(h(x)\in H\). For each agent \(a\in D\cup H\), let
\[
X_a:=\{x\in X:a\in\{d(x),h(x)\}\}.
\]
An allocation is a set \(Y\subseteq X\), and \(Y_a:=Y\cap X_a\). For
\(Z\subseteq X\), write
\[
d(Z):=\{d(x):x\in Z\}.
\]

Each agent \(a\in D\cup H\) is represented by a choice function
\[
C_a:2^{X_a}\longrightarrow 2^{X_a},
\qquad
C_a(S)\subseteq S.
\]
For \(Y\subseteq X\), write \(C_a(Y):=C_a(Y_a)\), and define
\[
C_D(Y):=\bigcup_{d\in D}C_d(Y),
\qquad
C_H(Y):=\bigcup_{h\in H}C_h(Y).
\]

A choice function \(C_a\) satisfies \emph{irrelevance of rejected contracts}
(IRC) if, for all \(Y\subseteq Z\subseteq X_a\),
\[
C_a(Z)\subseteq Y\subseteq Z
\quad\Longrightarrow\quad
C_a(Y)=C_a(Z).
\]
IRC is the standard consistency condition emphasized by
\citet{Aygun2013c}. In particular, it implies idempotence:
\[
C_a(C_a(Y))=C_a(Y).
\]

Following \citet{Kelso1982a}, a choice function \(C_a\) is
\emph{substitutable} if, for all \(Y\subseteq Z\subseteq X_a\),
\[
C_a(Z)\cap Y\subseteq C_a(Y).
\]
This is the standard contract-level substitutability condition used in
matching with contracts \citep{Hatfield2005a}.

A choice function is \emph{unitary} if no choice set contains two distinct
contracts involving the same doctor--hospital pair. Unitarity is weaker than
unit demand: it permits an agent to maintain several bilateral relationships
while treating alternative contractual specifications of a given
relationship as mutually exclusive.

Following \citet{Hatfield2010}, a hospital choice function \(C_h\) is
\emph{unilaterally substitutable} if, for all \(x,z\in X_h\) and
\(Y\subseteq X_h\) such that \(d(z)\notin d(Y)\),
\[
z\in C_h(Y\cup\{x,z\})
\quad\Longrightarrow\quad
z\in C_h(Y\cup\{z\}).
\]

\begin{assumption}
\label{ass:main}
Every doctor choice function is substitutable, unitary, and satisfies IRC.
Every hospital choice function is unilaterally substitutable, unitary, and
satisfies IRC.
\end{assumption}

Assumption~\ref{ass:main} separates partner choice from contract-term choice.
A hospital may revise the terms associated with a doctor, but complementarities
across doctor identities are restricted. Thus, the assumption permits
contract-level flexibility within a bilateral relationship without imposing
full substitutability across individual contracts.

Bilateral unitarity is substantive rather than merely technical.
\citet[Example~3]{Bando2026} construct a nonunitary many-to-many market in
which hospital choice satisfies unilateral substitutability and
substitutability across doctors, and the cumulative offer process has an
order-independent stable outcome, but no doctor-optimal stable allocation
exists. Our domain retains portfolio choice across partners while excluding
multiple simultaneous contracts within a given bilateral relationship.

Doctor substitutability and IRC imply path independence:
\[
C_d(Y\cup Z)=C_d(C_d(Y)\cup Z)
\qquad
\text{for all }d\in D\text{ and }Y,Z\subseteq X.
\]
We use this property repeatedly. Hospital choice functions need not be
substitutable or path independent.

For nonunitary many-to-many markets, \citet{Bando2026} introduce
\emph{substitutability across doctors} as an identity-level extension of
unilateral substitutability. A hospital choice function \(C_h\) satisfies
substitutability across doctors if, for all
\(Y\subseteq Z\subseteq X_h\) and \(d\in D\),
\[
\bigl|(C_h(Z)\cap Y)_d\bigr|
\leq
\bigl|C_h(Y)_d\bigr|.
\]
The condition requires a hospital that chooses contracts with a doctor from
a larger menu to continue choosing at least as many contracts with that
doctor when the menu shrinks and the relevant contracts remain available.

The following equivalence is due to \citet{Bando2026}. Under unitarity, the
two cardinalities above are either zero or one, so substitutability across
doctors reduces to preservation of doctor identities.

\begin{proposition}[Unilateral substitutability and substitutability across
doctors]
\label{prop:US-SAD}
Suppose that \(C_h\) is unitary. Then \(C_h\) is unilaterally substitutable
if and only if it satisfies substitutability across doctors.
\end{proposition}

For completeness, Appendix~\ref{app:US-SAD} provides a proof in the notation
of the present paper.
 
\subsection{Stability and revealed optimality}

An allocation \(Y\) is \emph{individually rational} if
\[
C_D(Y)=C_H(Y)=Y.
\]
A nonempty set \(Z\subseteq X\setminus Y\) \emph{blocks} \(Y\) if
\[
Z\subseteq C_D(Y\cup Z)\cap C_H(Y\cup Z).
\]
An allocation is \emph{stable} if it is individually rational and admits no
blocking set. Let \(\mathcal S\) denote the set of stable allocations.

For doctor-individually rational allocations \(Y\) and \(Y'\), write
\[
Y\succeq_D^B Y'
\quad\Longleftrightarrow\quad
C_D(Y\cup Y')=Y.
\]
Path independence makes \(\succeq_D^B\) a partial order, the doctor Blair
order, on the doctor-individually rational allocations \citep{Blair1988}.

A stable allocation \(Y^D\) is \emph{doctor-optimal} if
\[
Y^D\succeq_D^B Y
\qquad
\text{for every }Y\in\mathcal S.
\]
Equivalently,
\[
C_D(Y^D\cup Y)=Y^D
\qquad
\text{for every }Y\in\mathcal S.
\]
Thus, a doctor-optimal stable allocation is the greatest element of
\((\mathcal S,\succeq_D^B)\) and is unique whenever it exists. If doctor
choices are induced by strict preferences over bundles, this definition
coincides with the usual welfare notion of doctor optimality.

A stable allocation \(Y^P\) is \emph{hospital-pessimal} if
\[
C_H(Y^P\cup Y)=Y
\qquad
\text{for every }Y\in\mathcal S.
\]
This pairwise definition does not require hospital choices to generate a
partial order. When hospital choices are induced by strict preferences over
bundles, it coincides with the usual welfare notion of hospital pessimality.

For an allocation \(Y\), define
\[
\Gamma(Y):=\{x\in X:x\in C_D(Y\cup\{x\})\}.
\]
If \(C_D(Y)=Y\), then \(Y\subseteq\Gamma(Y)\). The following
characterization is used throughout the paper.

\begin{proposition}[Stability characterization]
\label{prop:stability-characterization}
Suppose that doctor choice functions are substitutable, hospital choice
functions are unitary, and all choice functions satisfy IRC. An allocation
\(Y\) is stable if and only if
\[
Y=C_H(\Gamma(Y)).
\]
\end{proposition}

The proof is in Appendix~\ref{app:stability-characterization}.

Unless stated otherwise, all subsequent results are established under
Assumption~\ref{ass:main}. The law of aggregate demand is imposed only in
Corollary~\ref{cor:rural-hospitals}.

\section{The doctor-proposing cumulative offer process}
\label{sec:COP}

\subsection{The process}

Set
\[
O^0=R^0=Y^0=\emptyset,
\]
where \(O^t\) is the set of contracts proposed during the first \(t\) steps,
\(R^t\) is the cumulative set of rejected contracts, and \(Y^t\) is the
tentative allocation. At state \(t\geq0\), the process terminates if
\[
C_D(X\setminus R^t)\subseteq O^t.
\]
Otherwise, choose a doctor \(d\) and a contract
\[
x_{t+1}\in C_d(X\setminus R^t)\setminus O^t,
\]
and update
\[
O^{t+1}:=O^t\cup\{x_{t+1}\},
\qquad
Y^{t+1}:=C_H(O^{t+1}),
\qquad
R^{t+1}:=R^t\cup R_H(O^{t+1}),
\]
where \(R_H(O):=O\setminus C_H(O)\). Since each contract is proposed at most
once, every trajectory terminates. Let \(t^*\) denote the terminal state.

The principal difficulty is that unilateral substitutability does not make
hospital rejections monotone on arbitrary menus. The cumulative-offer menus
nevertheless have the required monotonicity.

\begin{lemma}[Irrevocability of hospital rejections]
\label{lem:rejection-irrevocability}
Along every trajectory,
\[
R_H(O^{t-1})\subseteq R_H(O^t)
\qquad
\text{for }t=1,\ldots,t^*.
\]
\end{lemma}

The proof is in Appendix~\ref{app:COP-proofs}. If a rejected contract were
later restored, substitutability across doctors would identify an earlier
replacement contract involving the same doctor. Doctor substitutability and
unitarity would then generate a restored contract with a strictly earlier
proposal date. Iteration yields an impossible infinite descent.

The following consequences isolate the remaining technical content of the
argument.

\begin{lemma}[Cumulative rejections]
\label{cor:cumulative-rejections}
Along every trajectory,
\[
R^t=R_H(O^t)
\qquad
\text{for }t=0,\ldots,t^*.
\]
\end{lemma}

\begin{lemma}[Terminal identities]
\label{lem:COP-terminal-identities}
The terminal state satisfies
\[
Y^{t^*}=C_D(X\setminus R^{t^*})
\qquad\text{and}\qquad
O^{t^*}=\Gamma(Y^{t^*}).
\]
\end{lemma}

\begin{lemma}[Stable-allocation invariant]
\label{lem:COP-stable-invariant}
For every stable allocation \(Y\),
\[
O^t\subseteq\Gamma(Y)
\qquad\text{and}\qquad
R^t\cap Y=\emptyset
\]
for all \(t=0,\ldots,t^*\).
\end{lemma}

The proofs are in Appendix~\ref{app:COP-proofs}.

\subsection{Doctor optimality}

\begin{theorem}[Doctor-optimality of cumulative offers]
\label{thm:COP-doctor-optimal}
Every doctor-proposing cumulative-offer trajectory terminates at the greatest
stable allocation under the doctor Blair order. Hence, its outcome is the
unique doctor-optimal stable allocation.
\end{theorem}

\begin{proof}
By Lemma~\ref{lem:COP-terminal-identities},
\[
Y^{t^*}=C_H(O^{t^*})=C_H(\Gamma(Y^{t^*})),
\]
so Proposition~\ref{prop:stability-characterization} implies that
\(Y^{t^*}\) is stable.

Fix \(Y\in\mathcal S\). Lemma~\ref{lem:COP-stable-invariant} gives
\(Y\subseteq X\setminus R^{t^*}\). By path independence and
Lemma~\ref{lem:COP-terminal-identities},
\begin{align*}
C_D(Y^{t^*}\cup Y)
&=C_D\bigl(C_D(X\setminus R^{t^*})\cup Y\bigr)\\
&=C_D(X\setminus R^{t^*})
=Y^{t^*}.
\end{align*}
Thus, \(Y^{t^*}\succeq_D^B Y\). The terminal allocation is therefore the
greatest stable allocation under the doctor Blair order and is doctor-optimal. Uniqueness follows from
antisymmetry of the Blair order.
\end{proof}

\begin{corollary}[Order independence]
\label{cor:COP-order-independence}
All doctor-proposing cumulative-offer trajectories terminate at the same
allocation.
\end{corollary}

\begin{proof}
A greatest element of a partial order is unique.
\end{proof}

The corollary permits proposals to arrive in any admissible sequential order.
The final stable assignment is therefore insensitive to proposal timing,
even though intermediate tentative assignments may differ.

Order independence alone does not imply doctor-optimality. In the
nonunitary market of \citet[Example~3]{Bando2026}, the cumulative offer
process reaches the same stable outcome for every order of proposals, but
the stable set has no doctor-optimal element.
Theorem~\ref{thm:COP-doctor-optimal} shows that bilateral unitarity
strengthens procedural invariance into an extremal selection result: every
trajectory reaches the greatest stable allocation under the doctor Blair
order.

\subsection{Hospital-side maximality}

Two distinct restrictions delimit the preceding result.
\citet[Example~3]{Bando2026} show that unilateral substitutability need not
guarantee doctor-optimal stability once unitarity is removed. The following
result concerns the complementary question: holding the unitary domain
fixed, can unilateral substitutability be weakened while preserving a
universal guarantee of doctor-optimal stability? The maximal-domain theorem
of \citet{Kasuya2021} gives a negative answer.

Let \(\mathfrak C^u\) denote the class of unitary hospital choice functions
satisfying IRC, and let \(\mathfrak C^{US,u}\) be its subclass satisfying
unilateral substitutability. A hospital-choice domain \(\mathfrak D\subseteq
\mathfrak C^u\) \emph{guarantees doctor-optimal stability} if every finite
market whose hospital choice functions belong to \(\mathfrak D\), together
with every profile of substitutable, unitary doctor choice functions
satisfying IRC, admits a doctor-optimal stable allocation.

\begin{corollary}[Maximality of unilateral substitutability]
\label{cor:US-maximality}
For the class of finite markets with at least two hospitals,
\(\mathfrak C^{US,u}\) is inclusion-maximal among the subclasses of
\(\mathfrak C^u\) that guarantee doctor-optimal stability.
\end{corollary}

\begin{proof}
Theorem~\ref{thm:COP-doctor-optimal} establishes the guarantee on
\(\mathfrak C^{US,u}\). Consider any
\(C_h\in\mathfrak C^u\setminus\mathfrak C^{US,u}\). By
\citet[Theorem~2]{Kasuya2021}, either \(C_h\) itself fails to guarantee the
existence of a doctor-optimal stable allocation, or it can be combined with
unit-demand choice functions of other hospitals so that even stability is not
guaranteed. In the first case, a dummy hospital whose choice is always empty
can be added without changing the failure. The doctor choice functions in
the construction are induced by unit-demand preferences, and the auxiliary
hospital choice functions are unit demand. They therefore belong,
respectively, to the substitutable, unitary, and IRC
doctor domain and to \(\mathfrak C^{US,u}\). Hence, adjoining
\(C_h\) to \(\mathfrak C^{US,u}\) destroys the guarantee.
\end{proof}

The corollary is a domain statement. It does not assert that every market
containing a non-unilaterally-substitutable hospital lacks a doctor-optimal
stable allocation. Rather, no such hospital choice function can be added to
the US domain while preserving existence uniformly over the remaining market
and doctor choice functions. On the embedded unit-demand doctor subdomain,
\citet[Theorem~1]{Kasuya2021} further shows that US is necessary and sufficient
within the observably substitutable-across-doctors domain.

\section{Weakly hospital-quasi-stable allocations}
\label{sec:weak-quasi-stability}

The construction below relies on bilateral unitarity. Substitutability
across doctors alone does not generate an analogous doctor-optimal boundary
in the nonunitary domain, as demonstrated by
\citet[Example~3]{Bando2026}. Our purpose is to identify the order structure
that becomes available once alternative contracts within each bilateral
relationship are mutually exclusive.

The order-theoretic argument adapts firm-quasi-stability to unilateral
substitutability. Firm-quasi-stability preserves individual contracts;
unilateral substitutability generally preserves only the identity of the
doctor represented in a hospital's choice. This motivates the following
relaxation.

\subsection{The relaxed domain}

\begin{definition}
\label{def:weak-HQS}
An allocation \(Y\) is \emph{weakly hospital-quasi-stable} if
\[
C_D(Y)=Y
\]
and
\[
d(Y_h)\subseteq d(C_h(\Gamma(Y)))
\qquad
\text{for every }h\in H.
\]
\end{definition}

Weak hospital-quasi-stability separates continuity of relationships from
continuity of terms. Doctors choose their current portfolios. Hospitals may
replace a current contract by another contract with the same doctor, but each
doctor currently employed by a hospital remains represented when the hospital
chooses from all doctor-acceptable contracts. The relaxed domain thus
describes relationship-preserving renegotiation rather than stability.

Let
\[
\mathcal Q^w:=\{Y\subseteq X:Y\text{ is weakly hospital-quasi-stable}\}.
\]
The empty allocation belongs to \(\mathcal Q^w\). Every stable allocation
also belongs to \(\mathcal Q^w\), because stability implies
\(Y=C_H(\Gamma(Y))\). Hence,
\[
\mathcal S\subseteq\mathcal Q^w.
\]
We order \(\mathcal Q^w\) by the doctor Blair order.

\subsection{Lattice structure and improvement}

\begin{theorem}[Weakly hospital-quasi-stable lattice]
\label{thm:weak-HQS-lattice}
The partially ordered set \((\mathcal Q^w,\succeq_D^B)\) is a finite lattice.
For every nonempty collection \(\{Y^i\}_{i\in I}\subseteq\mathcal Q^w\),
\[
\bigvee_{i\in I}Y^i
=
C_D\left(\bigcup_{i\in I}Y^i\right).
\]
\end{theorem}

The proof is in Appendix~\ref{app:WHQS-proofs}. Let
\[
Y^Q:=C_D\left(\bigcup_{Y\in\mathcal Q^w}Y\right)
\]
denote the greatest element of \(\mathcal Q^w\). The lattice orders
relationship-preserving states by the doctors' portfolio choices; it does not
impose a lattice structure on stable allocations.

For \(Y\in\mathcal Q^w\), define
\[
T(Y):=C_D(C_H(\Gamma(Y))).
\]
The operator allows hospitals to revise terms from all doctor-acceptable
contracts and then lets doctors reselect their portfolios.

\begin{proposition}[Weak-quasi-stable improvement]
\label{prop:T-improvement}
For every \(Y\in\mathcal Q^w\),
\[
T(Y)\in\mathcal Q^w
\qquad\text{and}\qquad
T(Y)\succeq_D^B Y.
\]
\end{proposition}

The proof is in Appendix~\ref{app:WHQS-proofs}. The proposition is a one-step
improvement result; no isotonicity of \(T\) is imposed.

\begin{proposition}[Fixed-point characterization]
\label{prop:T-fixed-points}
For every \(Y\in\mathcal Q^w\),
\[
T(Y)=Y
\quad\Longleftrightarrow\quad
Y\in\mathcal S.
\]
\end{proposition}

\begin{proof}
If \(Y\in\mathcal S\), then
\(C_H(\Gamma(Y))=Y=C_D(Y)\), so \(T(Y)=Y\).

Conversely, suppose \(T(Y)=Y\), and set \(S:=C_H(\Gamma(Y))\). Then
\(C_D(S)=Y\subseteq S\). If \(x\in S\setminus Y\), then
\(x\in\Gamma(Y)\), whereas IRC applied to
\[
C_D(S)=Y\subseteq Y\cup\{x\}\subseteq S
\]
gives \(C_D(Y\cup\{x\})=Y\), a contradiction. Hence \(S=Y\), and
Proposition~\ref{prop:stability-characterization} implies that \(Y\) is
stable.
\end{proof}

\begin{theorem}[Stability of the greatest weakly hospital-quasi-stable allocation]
\label{thm:greatest-weak-HQS-stable}
The greatest weakly hospital-quasi-stable allocation \(Y^Q\) is stable.
\end{theorem}

\begin{proof}
Proposition~\ref{prop:T-improvement} gives
\(T(Y^Q)\in\mathcal Q^w\) and \(T(Y^Q)\succeq_D^B Y^Q\). The maximality of
\(Y^Q\) and antisymmetry imply \(T(Y^Q)=Y^Q\). The result follows from
Proposition~\ref{prop:T-fixed-points}.
\end{proof}

\section{Equivalent characterizations and consequences}
\label{sec:equivalence}

\begin{corollary}[Equivalent characterizations]
\label{thm:equivalent-characterizations}
The following allocations coincide:
\begin{enumerate}
    \item the outcome of every doctor-proposing cumulative-offer trajectory;
    \item the greatest weakly hospital-quasi-stable allocation;
    \item the greatest stable allocation under the doctor Blair order;
    \item the unique doctor-optimal stable allocation.
\end{enumerate}
Their common value is denoted by \(Y^D\).
\end{corollary}

\begin{proof}
Theorem~\ref{thm:COP-doctor-optimal} identifies the cumulative-offer outcome
with the greatest stable allocation. By
Theorem~\ref{thm:greatest-weak-HQS-stable}, \(Y^Q\) is stable; since
\(\mathcal S\subseteq\mathcal Q^w\), it Blair-dominates every stable
allocation. Hence, \(Y^Q\) is the same greatest stable allocation. Since the
doctor Blair order is a partial order, this allocation is the unique
doctor-optimal stable allocation.
\end{proof}

This equality is precisely the conclusion that fails in the nonunitary
example of \citet{Bando2026}: order-independent cumulative offers and
stability survive there, but the stable set has no doctor-greatest element.
Bilateral unitarity restores the coincidence between procedural selection
and order-theoretic extremality.

\subsection{Hospital pessimality}

\begin{proposition}[Hospital pessimality]
\label{prop:hospital-pessimality}
For every stable allocation \(Y\),
\[
C_H(Y^D\cup Y)=Y.
\]
Consequently, \(Y^D\) is hospital-pessimal.
\end{proposition}

\begin{proof}
Fix \(Y\in\mathcal S\). Since \(C_D(Y^D\cup Y)=Y^D\), doctor
substitutability implies \(Y^D\subseteq\Gamma(Y)\). Individual rationality
also gives \(Y\subseteq\Gamma(Y)\). Stability yields
\[
C_H(\Gamma(Y))=Y
\subseteq
Y^D\cup Y
\subseteq
\Gamma(Y).
\]
IRC therefore gives
\[
C_H(Y^D\cup Y)=C_H(\Gamma(Y))=Y.
\]
\end{proof}

The proposition identifies opposing revealed-choice boundaries of the stable
set: the doctors' greatest stable allocation is hospital-pessimal in the
pairwise sense defined above. When choices are rationalizable, this has the
usual ordinal interpretation. It does not compare aggregate or cardinal
welfare across the two sides.

\subsection{The rural hospitals theorem}

Choice function \(C_a\) satisfies the \emph{law of aggregate demand} if
\[
Y\subseteq Z
\quad\Longrightarrow\quad
|C_a(Y)|\leq |C_a(Z)|.
\]

\begin{corollary}[Rural hospitals theorem]
\label{cor:rural-hospitals}
Suppose that every doctor and hospital satisfies the law of aggregate demand.
Then, for any \(Y,Y'\in\mathcal S\),
\[
|Y_d|=|Y'_d|
\quad\text{for every }d\in D,
\qquad
|Y_h|=|Y'_h|
\quad\text{for every }h\in H.
\]
\end{corollary}

\begin{proof}
Fix \(Y\in\mathcal S\). The defining doctor-optimality identity and the
law of aggregate demand give
\[
|Y_d|=|C_d(Y)|
\leq |C_d(Y^D\cup Y)|=|Y_d^D|
\qquad
\text{for every }d.
\]
Proposition~\ref{prop:hospital-pessimality} gives
\[
|Y_h^D|=|C_h(Y^D)|
\leq |C_h(Y^D\cup Y)|=|Y_h|
\qquad
\text{for every }h.
\]
Summing the first inequalities over doctors and the second over hospitals
shows \(|Y|\leq|Y^D|\leq|Y|\). Hence, all individual inequalities bind.
\end{proof}

Because the market is bilaterally unitary, these cardinalities count active
relationships rather than multiple contracts within the same relationship.
The corollary allows stable allocations to differ in partners and terms, but
not in the number of relationships held by any doctor or hospital.

\section{Conclusion}
\label{sec:conclusion}

This paper identifies bilateral unitarity as a sufficient boundary for
doctor-optimal stability under unilateral substitutability. Doctors and
hospitals may choose portfolios across several partners, but alternative
contracts within a given doctor--hospital relationship are mutually
exclusive. On this domain, every doctor-proposing cumulative-offer trajectory
terminates at the greatest stable allocation under the doctor Blair order.

The role of unitarity is highlighted by
\citet[Example~3]{Bando2026}. In their nonunitary market, unilateral
substitutability and an order-independent cumulative offer process coexist
with the failure of doctor-optimal stability. Our result shows that excluding
multiple simultaneous contracts within a bilateral relationship restores the
doctor-optimal conclusion without returning to unit demand. Together with the
maximal-domain necessity result of \citet{Kasuya2021}, this establishes a
sharp hospital-side boundary: doctor demand can be expanded to substitutable
portfolios, whereas unilateral substitutability cannot generally be weakened
within the unitary domain.

The same allocation is the greatest element of the weakly
hospital-quasi-stable lattice. This characterization replaces contract
retention with relationship retention and recovers doctor-optimal stability
without requiring the stable set itself to form a lattice. The allocation is
also hospital-pessimal in the revealed-choice sense. Under the law of
aggregate demand, every agent maintains the same number of bilateral
relationships at all stable allocations.

\appendix

\section{Unilateral substitutability under unitarity}
\label{app:US-SAD}

The equivalence in Proposition~\ref{prop:US-SAD} is due to
\citet{Bando2026}. We reproduce a proof for completeness and to express the
argument in the notation used here.

\begin{proof}[Proof of Proposition~\ref{prop:US-SAD}]
Under hospital unitarity, substitutability across doctors is equivalent to
\begin{equation}
\label{eq:SAD-appendix}
d(C_h(Z)\cap Y)
\subseteq
d(C_h(Y))
\qquad
\text{for all }Y\subseteq Z\subseteq X_h.
\end{equation}

Suppose first that \(C_h\) is unilaterally substitutable. Let
\(Y\subseteq Z\subseteq X_h\). Suppose, contrary to
\eqref{eq:SAD-appendix}, that there exists
\[
z\in C_h(Z)\cap Y
\]
such that
\[
d(z)\notin d(C_h(Y)).
\]
Set
\[
Y^0:=C_h(Y).
\]
Since
\[
C_h(Y)\subseteq Y^0\cup\{z\}\subseteq Y,
\]
IRC implies
\[
C_h(Y^0\cup\{z\})=Y^0.
\]
In particular,
\[
z\notin C_h(Y^0\cup\{z\}).
\]

Enumerate
\[
\{x\in Z:d(x)\neq d(z)\}\setminus Y^0
=
\{x_1,\ldots,x_r\},
\]
and define
\[
Y^i:=Y^{i-1}\cup\{x_i\},
\qquad
i=1,\ldots,r.
\]
Because \(d(z)\notin d(Y^0)\),
\[
Y^r=\{x\in Z:d(x)\neq d(z)\}.
\]

Hospital unitarity and \(z\in C_h(Z)\) imply that \(z\) is the only
contract involving \(d(z)\) in \(C_h(Z)\). Hence,
\[
C_h(Z)\subseteq Y^r\cup\{z\}\subseteq Z.
\]
By IRC,
\[
C_h(Y^r\cup\{z\})=C_h(Z),
\]
so
\[
z\in C_h(Y^r\cup\{z\}).
\]

On the other hand, unilateral substitutability implies, by contraposition,
that
\[
z\notin C_h(Y^{i-1}\cup\{z\})
\quad\Longrightarrow\quad
z\notin C_h(Y^i\cup\{z\})
\]
for every \(i=1,\ldots,r\), because
\(d(z)\notin d(Y^{i-1})\). Starting from
\(z\notin C_h(Y^0\cup\{z\})\), induction gives
\[
z\notin C_h(Y^r\cup\{z\}),
\]
a contradiction. Therefore, \eqref{eq:SAD-appendix} holds, and \(C_h\)
satisfies substitutability across doctors.

Conversely, suppose that \(C_h\) satisfies substitutability across doctors.
Let \(x,z\in X_h\) and \(Y\subseteq X_h\) satisfy
\[
d(z)\notin d(Y)
\qquad\text{and}\qquad
z\in C_h(Y\cup\{x,z\}).
\]
Applying substitutability across doctors to
\[
Y\cup\{z\}\subseteq Y\cup\{x,z\}
\]
gives
\[
d(z)
\in
d\bigl(C_h(Y\cup\{x,z\})\cap(Y\cup\{z\})\bigr)
\subseteq
d(C_h(Y\cup\{z\})).
\]
Since \(z\) is the only contract involving doctor \(d(z)\) in
\(Y\cup\{z\}\), it follows that
\[
z\in C_h(Y\cup\{z\}).
\]
Thus, \(C_h\) is unilaterally substitutable.
\end{proof}

\section{A characterization of stability}
\label{app:stability-characterization}

\begin{proof}[Proof of Proposition~\ref{prop:stability-characterization}]
Suppose first that \(Y\) is stable. Individual rationality implies
\[
Y\subseteq\Gamma(Y).
\]
Define
\[
Z:=C_H(\Gamma(Y))\setminus Y.
\]
Fix \(h\in H\). Since
\[
C_h(\Gamma(Y))
\subseteq
Y_h\cup Z_h
\subseteq
\Gamma(Y)_h,
\]
IRC implies
\begin{equation}
\label{eq:appendix-hospital-block}
C_h(Y\cup Z_h)
=
C_h(\Gamma(Y)).
\end{equation}
Consequently,
\[
Z_h\subseteq C_h(Y\cup Z_h).
\]

We show that every contract in \(Z_h\) is also chosen by its doctor from
\(Y\cup Z_h\). Fix \(z\in Z_h\). Since
\[
Z_h\subseteq C_h(\Gamma(Y)),
\]
hospital unitarity implies
\[
Z_h\cap X_{d(z)}=\{z\}.
\]
Moreover, \(z\in Z_h\subseteq\Gamma(Y)\), so
\[
z\in C_{d(z)}(Y\cup\{z\}).
\]
Because \(z\) is the only contract in \(Z_h\) involving \(d(z)\),
the menus \(Y\cup\{z\}\) and \(Y\cup Z_h\) induce the same menu for
doctor \(d(z)\). Hence,
\[
z
\in
C_{d(z)}(Y\cup Z_h).
\]
It follows that
\[
Z_h
\subseteq
C_D(Y\cup Z_h)\cap C_H(Y\cup Z_h).
\]

If \(Z_h\neq\emptyset\), then \(Z_h\subseteq X\setminus Y\) blocks \(Y\),
contradicting stability. Therefore, \(Z_h=\emptyset\) for every \(h\), and
hence
\[
C_H(\Gamma(Y))\subseteq Y.
\]
Together with
\[
C_H(\Gamma(Y))\subseteq Y\subseteq\Gamma(Y)
\]
and \(C_H(Y)=Y\), IRC gives
\[
C_H(\Gamma(Y))=C_H(Y)=Y.
\]

Conversely, suppose that
\[
Y=C_H(\Gamma(Y)).
\]
Since a choice is contained in its menu,
\[
Y\subseteq\Gamma(Y).
\]
Therefore,
\[
C_H(\Gamma(Y))
\subseteq
Y
\subseteq
\Gamma(Y),
\]
and IRC gives
\[
C_H(Y)=C_H(\Gamma(Y))=Y.
\]

For every \(y\in Y\), the inclusion \(Y\subseteq\Gamma(Y)\) implies
\[
y\in C_{d(y)}(Y\cup\{y\})
=
C_{d(y)}(Y).
\]
Thus,
\[
Y\subseteq C_D(Y).
\]
The reverse inclusion is automatic, so
\[
C_D(Y)=Y.
\]
Hence, \(Y\) is individually rational.

Suppose, contrary to stability, that a nonempty set
\(Z\subseteq X\setminus Y\) blocks \(Y\). Then
\[
Z\subseteq C_D(Y\cup Z).
\]
For every \(z\in Z\), doctor substitutability applied to
\[
Y\cup\{z\}\subseteq Y\cup Z
\]
implies
\[
z\in C_D(Y\cup\{z\}).
\]
Therefore,
\[
Z\subseteq\Gamma(Y).
\]
It follows that
\[
C_H(\Gamma(Y))
=
Y
\subseteq
Y\cup Z
\subseteq
\Gamma(Y).
\]
By IRC,
\[
C_H(Y\cup Z)=C_H(\Gamma(Y))=Y.
\]
This contradicts the blocking condition
\[
Z\subseteq C_H(Y\cup Z),
\]
because \(Z\cap Y=\emptyset\) and \(Z\neq\emptyset\). Hence, no such
blocking set exists, and \(Y\) is stable.
\end{proof}

\section{Technical results for cumulative offers}
\label{app:COP-proofs}

\begin{proof}[Proof of Lemma~\ref{lem:rejection-irrevocability}]
By Proposition~\ref{prop:US-SAD}, every hospital satisfies
substitutability across doctors.

Suppose, contrary to the conclusion, that a contract is rejected by its
hospital and subsequently chosen again. Call such a contract
\emph{revived}. Choose a revived contract \(x_i\), where \(i\) is its
proposal date, and let
\[
h:=h(x_i).
\]
Let \(j\geq i\) be the first step at which \(h\) rejects \(x_i\). Thus,
\[
x_i\in C_h(O^s)
\qquad
\text{for }s=i,\ldots,j-1,
\]
while
\[
x_i\notin C_h(O^j).
\]
Since \(x_i\) is revived, there exists \(k>j\) such that
\[
x_i\in C_h(O^k).
\]

Because \(O^j\subseteq O^k\) and
\[
x_i\in C_h(O^k)\cap O^j,
\]
substitutability across doctors implies that \(C_h(O^j)\) contains a
contract \(x_\ell\) satisfying
\[
d(x_\ell)=d(x_i).
\]
Since \(x_i\notin C_h(O^j)\), we have \(x_\ell\neq x_i\). Moreover,
\(x_\ell\in O^j\), so
\[
\ell\leq j.
\]

We first show that
\[
\ell<i.
\]
Suppose instead that \(i<\ell\). Since \(j\) is the first rejection date
of \(x_i\) and \(\ell\leq j\), contract \(x_i\) has not been rejected
before Step \(\ell\). Hence,
\[
x_i\in X\setminus R^{\ell-1}.
\]
When proposed at Step \(i\),
\[
x_i\in C_{d(x_i)}(X\setminus R^{i-1}).
\]
Since
\[
X\setminus R^{\ell-1}
\subseteq
X\setminus R^{i-1},
\]
doctor substitutability implies
\[
x_i\in C_{d(x_i)}(X\setminus R^{\ell-1}).
\]
At Step \(\ell\), contract \(x_\ell\) is also chosen by the same doctor
from \(X\setminus R^{\ell-1}\). Thus,
\[
x_i,x_\ell
\in
C_{d(x_i)}(X\setminus R^{\ell-1}).
\]
The two contracts are distinct and involve the same doctor--hospital pair,
contradicting doctor unitarity. The case \(\ell=i\) is impossible because
\(x_\ell\neq x_i\). Therefore,
\[
\ell<i.
\]

We next show that \(x_\ell\) is itself revived. Suppose that
\[
x_\ell\notin R^{i-1}.
\]
Then \(x_\ell\in X\setminus R^{i-1}\). Since \(x_\ell\) was chosen by its
doctor when proposed at Step \(\ell<i\), and
\[
X\setminus R^{i-1}
\subseteq
X\setminus R^{\ell-1},
\]
doctor substitutability implies
\[
x_\ell\in C_{d(x_i)}(X\setminus R^{i-1}).
\]
At Step \(i\), the same doctor also chooses \(x_i\) from this menu.
Therefore,
\[
x_i,x_\ell
\in
C_{d(x_i)}(X\setminus R^{i-1}),
\]
again contradicting doctor unitarity. Hence,
\[
x_\ell\in R^{i-1}.
\]

Thus, \(x_\ell\) was rejected before Step \(i\). But
\[
x_\ell\in C_h(O^j),
\qquad
j\geq i,
\]
so \(x_\ell\) is subsequently chosen again by hospital \(h\). Therefore,
\(x_\ell\) is revived.

Starting from a revived contract with proposal date \(i\), we have
constructed another revived contract with strictly earlier proposal date
\(\ell<i\). Repeating the argument produces an infinite strictly decreasing
sequence of positive integers,
\[
i_1>i_2>i_3>\cdots,
\]
which is impossible. Hence, no contract can be revived.

Since \(O^{t-1}\subseteq O^t\), a contract in
\(R_H(O^{t-1})\setminus R_H(O^t)\) would be a rejected contract chosen
again at Step \(t\). As no contract is revived,
\[
R_H(O^{t-1})\subseteq R_H(O^t)
\qquad
\text{for every }t=1,\ldots,t^*.
\]
\end{proof}

\begin{proof}[Proof of Lemma~\ref{cor:cumulative-rejections}]
The conclusion holds at \(t=0\). If it holds at \(t-1\), then
Lemma~\ref{lem:rejection-irrevocability} gives
\[
R^t
=
R^{t-1}\cup R_H(O^t)
=
R_H(O^{t-1})\cup R_H(O^t)
=
R_H(O^t).
\]
\end{proof}

\begin{proof}[Proof of Lemma~\ref{lem:COP-terminal-identities}]
By the termination condition,
\[
C_D(X\setminus R^{t^*})\subseteq O^{t^*}.
\]
Since \(C_D(X\setminus R^{t^*})\) contains no rejected contract,
Corollary~\ref{cor:cumulative-rejections} gives
\[
C_D(X\setminus R^{t^*})
\subseteq
O^{t^*}\setminus R^{t^*}
=
C_H(O^{t^*})
=
Y^{t^*}
\subseteq
X\setminus R^{t^*}.
\]
By IRC,
\begin{equation}
\label{eq:COP-terminal-doctor-choice}
C_D(Y^{t^*})
=
C_D(X\setminus R^{t^*}).
\end{equation}

Let \(x_i\in Y^{t^*}\). Since \(x_i\) is not rejected by the terminal
step,
\[
Y^{t^*}\subseteq X\setminus R^{i-1}.
\]
Moreover,
\[
x_i\in C_{d(x_i)}(X\setminus R^{i-1})
\]
when it is proposed. Doctor substitutability therefore implies
\[
x_i\in C_{d(x_i)}(Y^{t^*}).
\]
Hence,
\[
Y^{t^*}\subseteq C_D(Y^{t^*}).
\]
The reverse inclusion follows from the definition of a choice function, so
\[
Y^{t^*}=C_D(Y^{t^*}).
\]
Together with \eqref{eq:COP-terminal-doctor-choice}, this proves
\[
Y^{t^*}=C_D(X\setminus R^{t^*}).
\]

We now prove that \(O^{t^*}=\Gamma(Y^{t^*})\). Let
\(x_i\in O^{t^*}\). Since
\[
Y^{t^*}\cup\{x_i\}
\subseteq
X\setminus R^{i-1}
\]
and \(x_i\) is chosen by its doctor from \(X\setminus R^{i-1}\),
substitutability gives
\[
x_i\in C_D(Y^{t^*}\cup\{x_i\}).
\]
Thus,
\[
O^{t^*}\subseteq\Gamma(Y^{t^*}).
\]

Conversely, suppose that
\[
z\in\Gamma(Y^{t^*})\setminus O^{t^*}.
\]
Every rejected contract has previously been proposed, so
\(R^{t^*}\subseteq O^{t^*}\). Hence,
\[
z\in X\setminus R^{t^*}.
\]
Using path independence and the first terminal identity,
\begin{align*}
C_D(Y^{t^*}\cup\{z\})
&=
C_D\bigl(C_D(X\setminus R^{t^*})\cup\{z\}\bigr)\\
&=
C_D\bigl((X\setminus R^{t^*})\cup\{z\}\bigr)\\
&=
C_D(X\setminus R^{t^*})\\
&=
Y^{t^*},
\end{align*}
where the third equality follows because
\(z\in X\setminus R^{t^*}\). This contradicts
\(z\in\Gamma(Y^{t^*})\). Therefore,
\[
O^{t^*}=\Gamma(Y^{t^*}).
\]
\end{proof}

\begin{proof}[Proof of Lemma~\ref{lem:COP-stable-invariant}]
We proceed by induction on \(t\). The conclusion is immediate at \(t=0\).

Suppose that
\[
O^{t-1}\subseteq\Gamma(Y)
\qquad\text{and}\qquad
R^{t-1}\cap Y=\emptyset.
\]
Then
\[
Y\subseteq X\setminus R^{t-1}.
\]
Since
\[
x_t\in C_{d(x_t)}(X\setminus R^{t-1})
\]
and
\[
Y\cup\{x_t\}\subseteq X\setminus R^{t-1},
\]
doctor substitutability implies
\[
x_t\in C_D(Y\cup\{x_t\}).
\]
Thus,
\[
O^t=O^{t-1}\cup\{x_t\}\subseteq\Gamma(Y).
\]

It remains to prove that \(R^t\cap Y=\emptyset\). By
Corollary~\ref{cor:cumulative-rejections}, it is enough to show
\[
O^t\cap Y\subseteq C_H(O^t).
\]
Let \(x_i\in O^t\cap Y\), and let \(h:=h(x_i)\). Stability and
Proposition~\ref{prop:stability-characterization} imply
\[
x_i\in Y_h=C_h(\Gamma(Y)).
\]
Since \(O^t\subseteq\Gamma(Y)\), substitutability across doctors implies that
\(C_h(O^t)\) contains a contract \(x_j\) with
\[
d(x_j)=d(x_i).
\]

Because \(x_j\in C_h(O^t)\), it is not in \(R^t\), and hence
\[
x_j\in X\setminus R^{i-1}.
\]
Contract \(x_i\) was chosen by its doctor when proposed at Step \(i\).
Doctor substitutability therefore gives
\[
x_i\in C_{d(x_i)}(\{x_i,x_j\}).
\]
On the other hand, \(x_j\in O^t\subseteq\Gamma(Y)\) and \(x_i\in Y\), so
\[
x_j\in C_{d(x_i)}(Y\cup\{x_j\}).
\]
Substitutability again implies
\[
x_j\in C_{d(x_i)}(\{x_i,x_j\}).
\]
Since the two contracts involve the same doctor--hospital pair, doctor
unitarity implies \(x_i=x_j\). Therefore,
\[
x_i\in C_h(O^t).
\]
It follows that \(R^t\cap Y=\emptyset\), completing the induction.
\end{proof}

\section{The weakly hospital-quasi-stable lattice}
\label{app:WHQS-proofs}

\begin{lemma}
\label{lem:Gamma-join}
Let \(\{Y^i\}_{i\in I}\subseteq\mathcal Q^w\) be nonempty and set
\[
Z:=C_D\left(\bigcup_{i\in I}Y^i\right).
\]
Then
\[
\Gamma(Z)\subseteq\bigcap_{i\in I}\Gamma(Y^i).
\]
\end{lemma}

\begin{proof}
Let
\[
U:=\bigcup_{i\in I}Y^i
\]
and take \(x\in\Gamma(Z)\). By path independence,
\[
x\in C_D(Z\cup\{x\})
=
C_D(C_D(U)\cup\{x\})
=
C_D(U\cup\{x\}).
\]
For every \(i\in I\),
\[
Y^i\cup\{x\}\subseteq U\cup\{x\}.
\]
Doctor substitutability therefore implies
\[
x\in C_D(Y^i\cup\{x\}),
\]
and hence \(x\in\Gamma(Y^i)\).
\end{proof}

\begin{proof}[Proof of Theorem~\ref{thm:weak-HQS-lattice}]
Let
\[
U:=\bigcup_{i\in I}Y^i
\qquad\text{and}\qquad
Z:=C_D(U).
\]
By idempotence of choice,
\[
C_D(Z)=Z.
\]

We show that \(Z\) is weakly hospital-quasi-stable. Fix \(h\in H\) and
\(d\in d(Z_h)\). Choose \(x\in Z_h\cap X_d\). Since \(Z\subseteq U\), there
exists \(j\in I\) such that \(x\in Y^j_h\). Because \(Y^j\in\mathcal Q^w\),
there exists
\[
x'\in C_h(\Gamma(Y^j))
\]
such that
\[
d(x')=d
\qquad\text{and}\qquad
h(x')=h.
\]

We claim that \(x'\in\Gamma(Z)\). Suppose otherwise. Then
\[
x'\notin C_d(Z\cup\{x'\}),
\]
so
\[
C_d(Z\cup\{x'\})\subseteq Z_d.
\]
Since \(C_d(Z)=Z_d\), IRC implies
\[
C_d(Z\cup\{x'\})=Z_d.
\]
Because \(x\in Z\subseteq\Gamma(Z)\), we also have \(x'\neq x\).
Thus,
\[
x\in C_d(Z\cup\{x'\}).
\]
By path independence,
\[
C_d(Z\cup\{x'\})
=
C_d(U\cup\{x'\}),
\]
and doctor substitutability gives
\[
x\in C_d(Y^j\cup\{x'\}).
\]
On the other hand, \(x'\in\Gamma(Y^j)\), so
\[
x'\in C_d(Y^j\cup\{x'\}).
\]
The two distinct contracts \(x\) and \(x'\) involve the same
doctor--hospital pair, contradicting doctor unitarity. Hence,
\(x'\in\Gamma(Z)\).

By Lemma~\ref{lem:Gamma-join},
\[
\Gamma(Z)\subseteq\Gamma(Y^j).
\]
Since
\[
x'\in C_h(\Gamma(Y^j))\cap\Gamma(Z),
\]
substitutability across doctors implies
\[
d\in d(C_h(\Gamma(Z))).
\]
Therefore,
\[
d(Z_h)\subseteq d(C_h(\Gamma(Z)))
\]
for every \(h\in H\), and \(Z\in\mathcal Q^w\).

We next verify that \(Z\) is the join. Since
\[
Z=C_D(U)\subseteq Z\cup Y^i\subseteq U,
\]
IRC gives
\[
C_D(Z\cup Y^i)=Z
\qquad
\text{for every }i\in I.
\]
Thus, \(Z\succeq_D^B Y^i\) for every \(i\).

Let \(\bar Y\in\mathcal Q^w\) be any other upper bound of
\(\{Y^i\}_{i\in I}\). Since the market is finite, only finitely many distinct
allocations occur in the collection. Repeated applications of path
independence and
\[
C_D(\bar Y\cup Y^i)=\bar Y
\]
give
\[
C_D(\bar Y\cup U)=\bar Y.
\]
Using path independence once more,
\[
C_D(\bar Y\cup Z)
=
C_D(\bar Y\cup C_D(U))
=
C_D(\bar Y\cup U)
=
\bar Y.
\]
Hence, \(\bar Y\succeq_D^B Z\), so \(Z\) is the least upper bound.

The empty allocation is the least element of \(\mathcal Q^w\). Since every
nonempty finite collection has a join, the meet of any collection is the join
of its common lower bounds. Therefore,
\((\mathcal Q^w,\succeq_D^B)\) is a finite lattice.
\end{proof}

\begin{lemma}
\label{lem:replacement-identity}
Let \(Y\in\mathcal Q^w\), set \(S:=C_H(\Gamma(Y))\), and fix a doctor
\(d\). Then, for every \(A\subseteq X_d\),
\[
C_d(S_d\cup A)=C_d(Y_d\cup S_d\cup A).
\]
\end{lemma}

\begin{proof}
Fix \(y\in S_d\). Since \(y\in\Gamma(Y)\),
\[
y\in C_d(Y_d\cup\{y\}).
\]
By substitutability,
\[
y\in
C_d\bigl((Y_{h(y)}\cap X_d)\cup\{y\}\bigr).
\]
All contracts in this menu involve the same doctor--hospital pair. Doctor
unitarity therefore implies
\[
C_d\bigl((Y_{h(y)}\cap X_d)\cup\{y\}\bigr)=\{y\}.
\]

Since \(Y\) is weakly hospital-quasi-stable, every hospital represented in
\(Y_d\) is also represented in \(S_d\). Consequently,
\[
Y_d\cup S_d
=
\bigcup_{y\in S_d}
\left((Y_{h(y)}\cap X_d)\cup\{y\}\right).
\]
Repeated applications of path independence yield
\begin{align*}
C_d(Y_d\cup S_d\cup A)
&=
C_d\left(
\bigcup_{y\in S_d}
C_d\bigl((Y_{h(y)}\cap X_d)\cup\{y\}\bigr)
\cup A
\right)\\
&=
C_d(S_d\cup A).
\end{align*}
\end{proof}

\begin{proof}[Proof of Proposition~\ref{prop:T-improvement}]
Fix \(Y\in\mathcal Q^w\), and write
\[
S:=C_H(\Gamma(Y))
\qquad\text{and}\qquad
Y':=T(Y)=C_D(S).
\]
By idempotence,
\[
C_D(Y')=Y'.
\]

We first show that
\[
\Gamma(Y')\subseteq\Gamma(Y).
\]
Take \(x\in\Gamma(Y')\), and let \(d=d(x)\). Path independence gives
\[
x\in C_d(Y'_d\cup\{x\})
=
C_d(S_d\cup\{x\}).
\]
By Lemma~\ref{lem:replacement-identity},
\[
x\in C_d(Y_d\cup S_d\cup\{x\}).
\]
Doctor substitutability then implies
\[
x\in C_d(Y_d\cup\{x\}),
\]
so \(x\in\Gamma(Y)\).

We next verify weak hospital-quasi-stability of \(Y'\). Fix \(h\in H\) and
\(d\in d(Y'_h)\). Choose \(x\in Y'_h\cap X_d\). Since \(Y'\subseteq S\),
\[
x\in C_h(\Gamma(Y)).
\]
Moreover, \(Y'\subseteq\Gamma(Y')\), so \(x\in\Gamma(Y')\). Since
\[
\Gamma(Y')\subseteq\Gamma(Y),
\]
substitutability across doctors implies
\[
d\in d(C_h(\Gamma(Y'))).
\]
Thus,
\[
d(Y'_h)\subseteq d(C_h(\Gamma(Y')))
\]
for every \(h\), and \(Y'\in\mathcal Q^w\).

Finally, for each doctor \(d\), path independence and
Lemma~\ref{lem:replacement-identity} imply
\begin{align*}
C_d(Y_d\cup Y'_d)
&=
C_d(Y_d\cup C_d(S_d))\\
&=
C_d(Y_d\cup S_d)\\
&=
C_d(S_d)\\
&=
Y'_d.
\end{align*}
Hence,
\[
C_D(Y\cup Y')=Y',
\]
and therefore \(Y'\succeq_D^B Y\).
\end{proof}

\bibliographystyle{econ}
\bibliography{dos_yang}

\end{document}